\documentclass[11pt]{article}
\usepackage{graphicx,amsmath,amssymb,amsthm,ifthen,fullpage}
\def\venue{arxiv}

\newcommand{\omitted}[1]{\ifthenelse{\equal{\venue}{arxiv}}{#1}{ The
    proof can be found in Ref.~\cite{BHH09} and is omitted from this
    extended abstract.}}

\newcommand{\be}{\begin{equation}}
\newcommand{\ee}{\end{equation}}
\newcommand{\ba}{\begin{array}}
\newcommand{\ea}{\end{array}}
\newcommand{\bea}{\begin{eqnarray}}
\newcommand{\eea}{\end{eqnarray}}
\def\bal#1\eal{\begin{align}#1\end{align}}
\def\bals#1\eals{\begin{align*}#1\end{align*}}

\DeclareMathOperator{\poly}{poly}
\DeclareMathOperator{\Var}{Var}
\def\bbE{\mathbb{E}}
\def\eps{\epsilon}
\def\l{\left}
\def\r{\right}

\newcommand{\calA}{{\cal A }}
\newcommand{\calD}{{\cal D }}

\newcommand{\calM}{{\cal M }}

\newcommand{\calO}{{\cal O }}

\newcommand{\CC}{\mathbb{C}}

\newcommand{\EE}{\mathbb{E}}
\newcommand{\la}{\langle}
\newcommand{\ra}{\rangle}

\newcommand{\prob}[1]{{\mathrm{Pr}}\left[#1\right]}

\newcommand{\expect}[1]{\mathbb{E}\left(#1\right)}
\newcommand{\EstProb}[3]{\mbox{\bf EstProb}({#1},{#2},{#3})}
\newcommand{\EstDist}[4]{\mbox{\bf EstDist}({#1},{#2},{#3},{#4})}
\newcommand{\UTest}[4]{\mbox{\bf UTest}({#1},{#2},{#3},{#4})}
\newcommand{\OTest}[4]{\mbox{\bf OTest}({#1},{#2},{#3},{#4})}
\newcommand{\bg}{\mathrm{Big}}
\newcommand{\wbg}{w_{\mathrm{big}}}

\newtheorem{dfn}{Definition}
\newtheorem{lemma}{Lemma}
\newtheorem{prop}{Proposition}
\newtheorem{theorem}{Theorem}
\newtheorem{problem}{Problem}
\newtheorem{corol}{Corollary}

\def\eq#1{(\ref{eq:#1})}

\def\lemref#1{Lemma~\ref{lem:#1}}

\title{Quantum algorithms for testing properties of distributions}

\author{Sergey Bravyi\footnote{IBM Watson Research Center, Yorktown Heights, NY 10598 (USA); {\em email:}
{\tt sbravyi@us.ibm.com}}, \ \  Aram W. Harrow\footnote{Department of
Mathematics, University of
Bristol, Bristol, U.K., {\em email:} {\tt a.harrow@bris.ac.uk}},  \ \ and  \ \
 Avinatan Hassidim\footnote{MIT, Cambridge, MA 02139 (USA);
{\em email:} avinatanh@gmail.com}}

\begin{document}

\maketitle

\begin{abstract}
Suppose one has access to oracles generating samples from two unknown probability distributions $p$ and $q$ on some $N$-element set.   How many samples does one need to test whether the two distributions are close or far from each other in the $L_1$-norm? This and related questions have been extensively studied during the last years in the field of property testing.
In the present paper we study  quantum algorithms for testing properties of distributions.
It is shown that  the $L_1$-distance $\|p-q\|_1$ can be estimated with a constant precision using only $O(N^{1/2})$ queries
in the quantum settings, whereas classical computers  need $\Omega(N^{1-o(1)})$ queries. We also describe quantum algorithms
for testing Uniformity and Orthogonality with query complexity $O(N^{1/3})$. The classical query complexity of these
problems is known to be $\Omega(N^{1/2})$.
A quantum algorithm for testing Uniformity has been recently independently discovered
by Chakraborty et al~\cite{CFMW09}.

\end{abstract}

%\tableofcontents

%\newpage

\section{Introduction}
\label{sec:intro}
\subsection{Problem statement and main results}
Suppose one has access to a black box generating independent samples from an
unknown probability distribution $p$
on some $N$-element set. If the number of available samples grows linearly with $N$, one can
use the standard Monte Carlo method to
simultaneously estimate the probability $p_i$ of every element
$i=1,\ldots,N$ and thus obtain a good approximation to the entire
distribution $p$.
On the other hand, many important questions that one usually encounters  in  statistical analysis
can be answered using only a {\em sublinear} number of samples.
For example, deciding whether $p$ is close in the $L_1$-norm to another distribution $q$
requires approximately $N^{1/2}$ samples if $q$ is known~\cite{Batu-indep-02} and approximately $N^{2/3}$ samples
if $q$ is also specified by a black-box~\cite{Batu-close-00}. Another example is estimating the Shannon entropy $H(p)=-\sum_i p_i \log_2{p_i}$. It was shown in~\cite{Batu-entropy-05,Val-symmetric-08} that distinguishing whether $H(p)\le a$ or $H(p)\ge b$ requires
approximately $N^{\frac{a}{b}}$ samples. Other examples include
deciding whether $p$ is close to a monotone or a unimodal distribution~\cite{Batu-monotone-04}, and deciding whether a pair of
distributions have disjoint supports~\cite{GR-bipartite-98}.
These and other questions fall into the field of {\em distribution testing}~\cite{Batu-thesis-01,Val-symmetric-08} that studies how many samples one needs
to decide whether an unknown distribution has a certain property or is far from having this property.
The purpose of the present paper is to explore whether quantum computers are capable of solving
distribution testing problems more efficiently.

The black-box sampling model adopted in~\cite{Batu-indep-02,Batu-close-00,Batu-entropy-05,Batu-monotone-04,Batu-thesis-01,Val-symmetric-08} assumes that a tester is presented with a list of samples drawn from an unknown distribution. What does it mean to sample from an unknown distribution in the quantum settings?
Let us start by casting  the black-box sampling model into a form that admits a quantum generalization.
Suppose $p$ is an unknown
distribution on an $N$-element set $[N]\equiv \{1,\ldots,N\}$ and  let $S$ be some specified integer.
We shall assume that $p$ is represented by an {\em oracle}  $O_p\, : \, [S]\to [N]$ such that
a probability $p_i$ of any element $i\in [N]$ is proportional to the number of elements in the pre-image of $i$, that is,
the number of  inputs $s\in [S]$ such that  $O_p(s)=i$.
In other words, one can sample from $p$ by querying the oracle $O_p$ on a random input  $s\in [S]$ drawn from the uniform distribution\footnote{Although in this model probabilities $p_i$ can only take values that are multiples of $1/S$,
choosing sufficiently large $S$ allows one to represent any distribution $p$ with an arbitrarily small error.}.
Note that a tester interacting with an oracle  can potentially be more powerful  due to the
possibility of making adaptive queries which could allow him to learn the internal structure of the oracle
as opposed to the black-box model.
However, it will be shown below (see Lemma~\ref{lemma:adapt} in Section~\ref{sec:lower})
that  the oracle  model
and the black-box model  are in fact equivalent. More precisely, for any fixed $N$ one can always choose
sufficiently large $S$ such that a tester will need the same number of queries in both models.

The oracle model admits a standard quantum generalization. Specifically, we shall transform
the oracle $O_p$ into a reversible form by keeping a copy of the input and writing the output of $O_p$ into
an ancillary register. A quantum oracle generating $p$ is a unitary operator whose action on basis vectors
coincides with the reversible version of $O_p$, as we will explain
further in Section~\ref{sec:prelim}.

The present paper focuses on testing three particular properties of distributions, namely,
{\em Statistical Difference}, {\em Orthogonality}, and {\em Uniformity}.
The corresponding property testing problems are promise problems so that a tester is required
to give a correct answer (with a bounded error probability) only for those instances that satisfy the promise.

\begin{problem}[\bf Testing Uniformity]${}$\\
{\rm \underline{Instance:}} Integers $N,S$, precision $\epsilon>0$. Access to an oracle generating a distribution $p$ on $[N]$. \\
{\rm \underline{Promise:}} \parbox[t]{15cm}{Either $p$ is the uniform distribution or the $L_1$-distance between $p$ and the uniform
distribution is at least $\epsilon$.}\\ \\
Decide which one is the case.
\end{problem}

\begin{problem}[\bf Testing Orthogonality]${}$\\
{\rm \underline{Instance:}} Integers $N,S$, precision $\epsilon>0$. Access to oracles generating distributions $p,q$ on $[N]$. \\
{\rm \underline{Promise:}} \parbox[t]{15cm}{Either $p$ and $q$ are orthogonal or  the $L_1$-distance between $p$ and $q$
is at most $2-\epsilon$.}\\
Decide which one is the case.
\end{problem}

\begin{problem}[\bf Testing Statistical Difference]${}$\\
{\em \underline{Instance:}} \parbox[t]{15cm}{Integers $N,S$, thresholds $0\le a<b\le 2$. Access to oracles generating distributions $p$ and $q$ on $[N]$.} \\
{\em \underline{Promise:}} \parbox[t]{15cm}{Either  $\|p-q\|_1\le a$ or $\|p-q\|_1\ge b$.} \\
Decide which one is the case.
\end{problem}

We assume that the precision $\epsilon$ is bounded from below by a fixed constant independent of $N$, for instance, $\epsilon\ge 1/10$. The same applies to the decision gap $b-a$ for testing Statistical Difference.
Given a function $f(N)$ we shall say that a property is testable in $f(N)$ queries if there exists a testing algorithm
making at most $f(N)$ queries that gives a correct answer with a sufficiently high probability (say $2/3$)
for any distributions $p$, $q$ satisfying the promise and for any oracles\footnote{Note that
 according to this definition a tester needs at most $f(N)$ queries even in the limit
$S\to \infty$. } specifying $p$ and $q$.
If a promise is violated, a tester can give an arbitrary answer.

Our main results are the following theorems.
\begin{theorem}
\label{thm:statdif}
Statistical Difference is testable on a quantum computer in $O(N^{1/2})$ queries.
\end{theorem}
\begin{theorem}
\label{thm:uniform}
Uniformity is testable on a quantum computer in $O(N^{1/3})$ queries.
\end{theorem}
\begin{theorem}
\label{thm:orthog}
Orthogonality is testable on a quantum computer in $O(N^{1/3})$ queries.
\end{theorem}

%SBB:
It is known that classically testing Orthogonality and Uniformity
requires $\Omega(N^{1/2})$ queries, see Sections~\ref{subs:collision}
and~\ref{subs:ulower}, while Statistical Difference is not testable in
$O(N^\alpha)$ queries for any $\alpha<1$, see~\cite{Val-symmetric-08}.
Therefore quantum computers provide a polynomial speedup for testing
Uniformity, Orthogonality, and Statistical Difference in terms of
query complexity.

Testing Orthogonality is closely related to the Collision Problem studied in~\cite{BHT-collision-97,Aar-collision-02}.
In Section~\ref{subs:collision} we describe a randomized reduction from the Collision Problem to testing Orthogonality.
Using the quantum lower bound for the Collision Problem due to Aaronson and Shi~\cite{AS-collision-04}
we obtain the following result.
%SBB:
\begin{theorem}
\label{thm:qlower}
Testing Orthogonality on a quantum computer requires $\Omega(N^{1/3})$ queries.
\end{theorem}

Quite recently Chakraborty, Fischer, Matsliah, and de Wolf~\cite{CFMW09} independently discovered a quantum Uniformity testing algorithm with query complexity $O(N^{1/3})$ and proved a lower bound $\Omega(N^{1/3})$ for testing Uniformity.
These authors also presented a quantum algorithm for testing whether
an unknown distribution $p$ coincides with a known distribution $q$ with query complexity $\tilde{O}(N^{1/3})$.

\subsection{Discussion and open problems}
 One motivation for studying distribution testing problems
 is that testing Orthogonality and Statistical Difference are complete problems for the complexity
class SZK (Statistical Zero Knowledge). More precisely, the following problem known as
{\em Statistical Difference} was shown to be
SZK-complete by Vadhan~\cite{SV-SZK-97}:\\
{\bf Input:} {\it  description of classical circuits $C_p,C_q$ that implement oracle functions $O_p,O_q\, : \, [S]\to [N]$ and
a pair of real numbers $0\le a<b\le 2$ such that $2a\le b^2$.}\\
{\bf Problem:} {\it Decide whether $\| p-q\|_1 \ge b$ (yes-instance) or $\|p-q\|_1\le a$ (no-instance)}.

\noindent
The class SZK includes many interesting algebraic and graph theoretic problems such as Discrete Logarithm, Graph Isomorphism, Graph NonIsomorphism, Quadratic Residuosity,  and The Shortest Vector in Lattice, see~\cite{AT-adia-03} and references therein. Thus it is  natural to ask whether quantum computers provide a universal speedup for problems in SZK similar to the square-root speedup for problems in NP provided by the Grover search algorithm.
Assuming that the circuits $C_p,C_q$ have size $poly(\log{(N)})$, one
can easily translate  the testing algorithm described in Section~\ref{sec:statdif}
to a quantum circuit of size $\tilde{O}(\sqrt{N})$ solving Statistical Difference problem
for any constants $a,b$ as above.
On the other hand, any classical algorithm treating the circuits $C_p,C_q$ as black boxes would
need roughly $N^{1-o(1)}$ queries, see~\cite{Val-symmetric-08}, thus requiring a circuit of size $\Omega(N^{1-o(1)})$.

Note that the Statistical Difference problem with $b=2$ is equivalent
to testing Orthogonality.  It can be solved classically in time
$\tilde{O}(N^{1/2})$ using the classical collision finding algorithm.
Unfortunately, the circuit complexity of the quantum Orthogonality
testing algorithm described in Section~\ref{sec:orthog} may be
different from its query complexity since it uses a quantum membership
oracle for a randomly generated set.  It is an open problem
whether Statistical Difference problem with $b=2$ can be solved by a
quantum circuit of size $\tilde{O}(N^{1/3})$, although with a suitably
powerful model of quantum RAM, such membership queries can be done in
time $\poly\log(N)$.  A related question is that of space-time
tradeoffs: our algorithms generally require storing $N^{O(1)}$
classical bits and then querying them with quantum algorithms that use
$\poly(\log(N)$ qubits.  We suspect that this amount of storage cannot
be reduced without increasing the run-time, but do not have a proof of
this conjecture.  Similar issues of quantum data structures for set
membership and conjectured space-time tradeoffs have arisen for the
element distinctness problem\cite{Amb-distinct-07,GR-memory-04}.

It is worth mentioning that all distribution properties studied in this paper
are {\em symmetric}, that is, these properties are invariant under relabeling of elements in the
underlying set $\{1,\ldots,N\}$. Testing  symmetric properties of distributions is equivalent to
testing properties of functions from $[S]$ to $[N]$ that are invariant under any permutations
of inputs and outputs of the function. It was recently shown by Aaronson and Ambainis that
quantum computers can provide at most polynomial speedup for testing
properties of such symmetric functions~\cite{AA-structure-09}.

More interesting than the mere fact of polynomial speedups provided by Theorems~\ref{thm:statdif},\ref{thm:uniform},\ref{thm:orthog}
is the way in which  our algorithms achieve it.  Classically, the results of
Ref.~\cite{Val-symmetric-08} provide  a simple characterization
of an asymptotically optimal testing algorithm for any symmetric property
of a distribution (satisfying certain natural continuity conditions).
By contrast, our algorithms use a variety of different strategies both to query the
oracles and to analyze the results of those queries.  These strategies
appear not to be special cases of the quantum walk framework which has
been responsible for most of the polynomial quantum speedups found to
date~\cite{Szegedy04,Santha08}.  A major challenge for future research
is to give a quantum version of Ref.~\cite{Val-symmetric-08}'s Canonical Tester algorithm;
in other words, we would like to characterize optimal quantum
algorithms for testing any symmetric property of a distribution (or a pair of distributions).

Finally, let us remark that the algorithm for estimating statistical difference described in Section~\ref{sec:statdif}
can be easily generalized to construct a quantum algorithm for estimating the von Neumann entropy of a black-box
distribution with query complexity $\tilde{O}(N^{1/2})$.
Using similar ideas one can construct an $\tilde{O}(N^{1/2})$-time
algorithm for estimating the fidelity between two black-box
distributions (i.e. $\sum_{i=1}^N \sqrt{p_iq_i}$).

The rest of the paper is organized as follows.
Section~\ref{sec:prelim} introduces necessary notations and basic facts about the quantum counting algorithm by Brassard,
Hoyer, Mosca, and  Tapp~\cite{BHMT-ampl-00}. The distribution testing algorithms described in the rest of the paper are actually classical probabilistic algorithms using the quantum counting as a subroutine.
 Theorem~\ref{thm:statdif} is proved in Section~\ref{sec:statdif}.
Theorem~\ref{thm:uniform} is proved in Section~\ref{sec:uniform}. Theorem~\ref{thm:orthog} is proved
in Section~\ref{sec:orthog}. We discuss lower bounds for the above distribution testing problems in Section~\ref{sec:lower}.

%Implications, significance, open problems. Uniformity: mixing time of random walks, reduction from the Statistical Difference %with
%known q to testing Uniformity. Orthogonality: SZK-complete, graph non-isomorphism as an example,
 %testing bipartiteness. Statistical Difference: ?.

%%%%%%%%%%%%%%%%%%%%%%%%%%%%%%%%%%%%%
%%%%%%%%%%%%%%%%%%%%%%%%%%%%%%%%%%%%%
\section{Preliminaries}
\label{sec:prelim}
%%%%%%%%%%%%%%%%%%%%%%%%%%%%%%%%%%%%%
%%%%%%%%%%%%%%%%%%%%%%%%%%%%%%%%%%%%%
Let $\calD_N$  be a set of probability
distributions $p=(p_1,\ldots,p_N)$ such that a probability $p_i$ of any element $i\in [N]$ is a rational number.
Let us say that an oracle $O\, : \, [S]\to [N]$ generates a distribution $p\in \calD_N$ iff  for all $i\in [N]$
the probability $p_i$ equals the fraction of inputs $s\in [S]$ such that $O(s)=i$,
\[
p_i=\frac1S\,  \# \{ s\in [S]\, : \, O(s)=i\}.
\]
Note that the identity of elements in the domain of an oracle  $O$ is
irrelevant, so if $O$ generates $p$ and $\sigma$ is any
permutation on $[S]$ then $O\circ \sigma$ also generates $p$. By definition, any map $O\, : \, [S]\to [N]$
generates some distribution $p\in \calD_N$.

For any oracle $O\, : \, [S]\to [N]$ we shall define a quantum oracle $\hat{O}$ by transforming $O$ into a reversible form
and allowing it to accept coherent superpositions of queries.
Specifically, a quantum oracle $\hat{O}$ is a unitary operator  acting on a Hilbert space $\CC^S \otimes \CC^{N+1}$
equipped with a standard basis $\{ |s\ra\otimes |i\ra\}$, $s\in [S]$, $i\in \{0\}\cup [N]$
such that
\be
\hat{O}\, |s\ra \otimes |0\ra = |s\ra \otimes |O(s)\ra \quad \mbox{for all $s\in [S]$}.
\ee
In other words, querying $\hat{O}$ on a basis vector $|s\ra\otimes |0\ra$
one gets the output of the classical oracle $O(s)$ in the second register while the first register keeps a copy of $s$
to maintain unitarity.
The action of $\hat{O}$ on a subspace in which the second register is orthogonal to the state $|0\ra$
can be arbitrary. We shall assume that a quantum tester can execute operators $\hat{O}$, $\hat{O}^\dag$
and the controlled versions of them. Execution of any one of these operators counts as one query.

We shall see that all testing problems posed in Section~\ref{sec:intro} can be reduced (via classical randomized reductions) to the following problem.
\begin{problem}[\bf Probability Estimation]
Given integers $S,N$, description of a subset $A\subset [N]$,  precision $\delta$,
error probability $\omega$,  and access to an oracle
generating some distribution $p\in \calD_N$. Let $p_A=\sum_{i\in A} p_i$ be the total probability of $A$.
One needs to generate an estimate $\tilde{p}_A$ satisfying
\be
\label{approx}
\prob{|\tilde{p}_A - p_A |\le \delta}\ge 1-\omega.
\ee
\end{problem}

Our main technical tool will be the quantum counting algorithm  by Brassard et al.~\cite{BHMT-ampl-00}.
Specifically, we shall use the following version of Theorem~12 from~\cite{BHMT-ampl-00}.
\begin{theorem}
\label{thm:BHMT}
There exists a quantum algorithm $\EstProb{p}{A}{M}$
taking as input a distribution $p\in \calD_N$ specified by an oracle,
a subset $A\subset [N]$, and an integer $M$.
The algorithm makes exactly $M$ queries to
the oracle generating $p$ and outputs an estimate $\tilde{p}_A$ such that
\be
\label{BHMT}
\prob{|\tilde{p}_A - p_A|\le \delta}\ge 1-\omega
\ee
for all $\delta>0$ and $0\le \omega \le 1/2$ satisfying
\be
\label{BHMT1}
M\ge \frac{c\sqrt{p_A}}{\omega \delta} \quad \mbox{and} \quad M\ge \frac{c}{\omega \sqrt{\delta}}.
\ee
Here $c=O(1)$ is some constant.
If $p_A=0$ then $\tilde{p}_A=0$ with certainty.
\end{theorem}

\omitted{\begin{proof}
Let $O\, : \, [S]\to [N]$ be the oracle generating $p$.
Using one query to $\hat{O}$ and one query to $\hat{O}^\dag$ one can implement a phase-flip oracle
$W_A\, : \, \CC^S \to \CC^S$ such that
\[
W_A\, |s\ra = \left\{ \ba{rcl}
-|s\ra &\mbox{if} & O(s)\in A, \\
|s\ra &\mbox{if} & O(s) \notin A. \\
\ea \right.
\]
Theorem~12 from~\cite{BHMT-ampl-00} implies that for any integer $M'\ge 1$ there exists a quantum algorithm
using an operator $\Lambda(W_A)$ exactly $M'$ times that outputs an estimate $\tilde{p}_A$ ($0\le \tilde{p}_A \le 1$) satisfying
\be
\label{phase_est1}
\prob{|\tilde{p}_A - p_A| \le 2\pi k \frac{\sqrt{p_A (1-p_A)}}{M'} + k^2 \frac{\pi^2}{(M')^2}}\ge 1-\frac1{2(k-1)}
\ee
for all integers $k\ge 2$.  Moreover, if $p_A=0$ then $\tilde{p}_A=0$ with certainty.

Choosing $k$ as the smallest integer such that $k\ge 1+1/2\omega$ and
$M=2M'$ we conclude that Eq.~(\ref{BHMT1}) holds whenever
\[
\frac{\sqrt{p_A}}{\omega M} \le c' \delta  \quad \mbox{and} \quad
\frac{1}{\omega^2 M^2} \le c'' \delta
\]
for some constants $c',c''$. This is equivalent to Eq.~(\ref{BHMT1}).

\end{proof}}

%%%%%%%%%%%%%%%%%%%%%%%%%%%%%%%%%%%%%%%%%%%%%%%%%%
%%%%%%%%%%%%%%%%%%%%%%%%%%%%%%%%%%%%%%%%%%%%%%%%%%
\section{Quantum algorithm for estimating statistical difference}
\label{sec:statdif}
%%%%%%%%%%%%%%%%%%%%%%%%%%%%%%%%%%%%%%%%%%%%%%%%%%
%%%%%%%%%%%%%%%%%%%%%%%%%%%%%%%%%%%%%%%%%%%%%%%%%%
In this section we prove Theorem~\ref{thm:statdif}. Let $p,q\in \calD_N$ be unknown distributions specified by
oracles. Define an auxiliary distribution $r\in \calD_N$ such that
$r_i=(p_i + q_i)/2$ for all $i\in [N]$. If we can sample $i$ from both $p$ and $q$ then by choosing randomly between these
two options we can also sample $i$ from $r$.
Let $x\in [0,1]$ be a random variable which takes value
\[
x_i=\frac{|p_i-q_i|}{p_i+q_i}
\]
with probability $r_i$. It is evident that
\be
\label{EE(x)}
\EE(x)=\sum_{i\in [N]} r_i x_i = \frac12  \sum_{i\in [N]} |p_i -q_i| = \frac12\, \| p-q \|_1.
\ee
Thus in order to estimate the distance $\| p-q \|_1$  it suffices to estimate
the expectation value $\EE(x)$ which can be done using the standard Monte Carlo method.
Since we have to estimate $\EE(x)$  only with a constant precision, it
suffices to generate $O(1)$ samples of $x_i$.
Given a sample of $i$ (which is easy to generate classically) we can estimate $x_i$ by calling the probability estimation
algorithm to get estimates of $p_i$ and $q_i$.
It suggests the following algorithm for estimating the distance $\| p-q \|_1$.
\begin{center}
%\begin{figure}[h]
\fbox{\parbox{16cm}{ $\EstDist{p}{q}{\epsilon}{\tau}$
\noindent

Set $n=27/\tau\epsilon^2$, $M=c \sqrt{N}/\epsilon^6 \tau^4$.

Let $i_1,\ldots,i_n\in [N]$ be a list of $n$ independent samples drawn from $r$.

For $a=1,\ldots,n$

\{

$\quad$
\parbox[t]{15cm}{
Let $\tilde{p}_{i_a}$ be estimate of $p_{i_a}$ obtained using $\EstProb{p}{\{i_a\}}{M}$.

Let $\tilde{q}_{i_a}$ be estimate of $q_{i_a}$ obtained using  $\EstProb{q}{\{i_a\}}{M}$.

Let $\tilde{x}_{i_a}= |\tilde{p}_{i_a} - \tilde{q}_{i_a}|/( \tilde{p}_{i_a} + \tilde{q}_{i_a})$ be estimate of $x_{i_a}$.

}

\}

Output $\tilde{x}=(1/n)\sum_{a=1}^n \tilde{x}_{i_a}$.}}
%\end{figure}
\end{center}
Here $c=O(1)$ is a constant whose precise value will not be important for us.

\begin{lemma}
\label{lemma:statdif}
The algorithm $\EstDist{p}{q}{\epsilon}{\tau}$
outputs an estimate $\tilde{x}$ satisfying
\be
\prob{ | \tilde{x} -\EE(x) |< \epsilon}\ge 1-\tau,
\ee
where $\EE(x)=(1/2)\|p-q\|_1$.
\end{lemma}
\begin{proof}
Define a random variable
\[
\bar{x}=\frac1n\sum_{a=1}^n x_{i_a},
\]
where $i_1,\ldots,i_n$ is a list of samples generated at the first step of the algorithm.
Note that $\EE(\bar{x})=\EE(x)$ and $\Var{(\bar{x})}=\Var{(x)}/n$.
 As $|p_i-q_i|\le p_i+q_i$ we have $0\le x_i\le 1$ and so one can bound the variance of $x$ as
$\Var{(x)} \le \EE(x^2) \le 1$.
Therefore  $\Var{(\bar{x})}\le 1/n$. Applying  the Chebyshev inequality to $\bar{x}$ one gets
\be
\label{error_bound1}
\prob{|\bar{x}-\EE(x)|\ge \epsilon/3}\le \frac{9\Var{(\bar{x})}}{\epsilon^2} \le \frac{9}{n\epsilon^2}\le \frac{\tau}{3}.
\ee
Let $\tilde{x}$ be the output of $\EstDist{p}{q}{\epsilon}{\tau}$. The union bound implies that
\be
\label{error_bound2}
\prob{|\tilde{x}-\bar{x}|\ge \epsilon/3} \le \prob{\exists a\, : \, |\tilde{x}_{i_a} - x_{i_a}|\ge \epsilon/3n}\le n
\prob{|\tilde{x}_{i} - x_{i}|\ge \epsilon/3n},
\ee
where $i\equiv i_a$ is a sample drawn from $r$.
Therefore it suffices to verify that
\be
\label{error_bound3}
\prob{|\tilde{x}_{i} - x_{i}|\ge \epsilon/3n} \le \frac{2\tau}{3n}.
\ee
Let us say that an element $i$ is {\em bad} iff
\be
\label{bad_element}
\max{(p_i,q_i)} \le \frac{\tau}{3nN} \quad (\mbox{bad element}).
\ee
The probability that $i$ is bad is at most
\[
p_{bad}=\sum_{i \; {\mathrm is \; bad}} r_i \le \frac{\tau}{3n}.
\]
Therefore it suffices to get a bound
\be
\label{error_bound4}
\prob{|\tilde{x}_{i} - x_{i}|\ge \epsilon/3n\,  |\,  i \; \mbox{is good}}\le \frac{\tau}{3n},
\ee
where we conditioned on $i$ being a good (not bad) element.

Let us translate the precision up to which one needs to estimate $x_i$ into a precision up to which one needs to
estimate $p_i$ and $q_i$.
\begin{prop}
Consider a real-valued function $f(p,q)=(p-q)/(p+q)$ where $0\le p,q\le 1$. Assume that
 $|p-\tilde{p}|, |q-\tilde{q}|\le \delta(p+q)$  for some $\delta\le 1/5$. Then
\be
|f(p,q)-f(\tilde{p},\tilde{q})|\le {5\delta}.
\ee
\end{prop}
\omitted{
\begin{proof}
Assume without loss of generality that $p\ge q$. Computing the partial derivatives of $f(p,q)$ one gets
\[
\partial_p f(p,q)= \frac{2q}{(p+q)^2}, \quad \partial_q f(p,q) =-\frac{2p}{(p+q)^2}
\]
both of which have absolute value at most $2/(p+q)$.It follows that
\[
|f(p,q)-f(\tilde{p},\tilde{q})| \le \frac2{\min{\{ p+q,\tilde{p}+\tilde{q}\}}}\, \left( |p-\tilde{p} | + |q-\tilde{q}|\right).
\]
The condition of the lemma implies that $\tilde{p}+\tilde{q}\ge (p+1)(1-\delta)$, so that
\[
|f(p,q)-f(\tilde{p},\tilde{q})|  \le \frac{4\delta}{1-\delta} \le 5\delta.
\]
\end{proof}}
Note that
\[
|\tilde{x}_i - x_i| = |\, |f(\tilde{p}_i,\tilde{q}_i)| - | f(p_i,q_i)||\le | f(\tilde{p}_i,\tilde{q}_i) -  f(p_i,q_i) |.
\]
Since we want to estimate $x_i$ with a precision $\epsilon/3n$, it suffices to estimate $p_i$ and $q_i$
with a precision $\delta (p_i+q_i)\ge \delta \max{(p_i,q_i)}$ where $5\delta=\epsilon/3n$, that is, $\delta=\epsilon/(15n)$.
Summarizing,
\be
\label{error_bound5}
|\tilde{p}_i - p_i|, |\tilde{q}_i - q_i| \le \frac{\epsilon}{15n}\max{(p_i,q_i)} \quad \Rightarrow \quad
|\tilde{x}_i - x_i|\le \frac{\epsilon}{3n}.
\ee
 Thus  it suffices to estimate
$p_i$ and $q_i$ with precision
\be
\label{delta}
\delta \sim \epsilon n^{-1} \max{(p_i,q_i)} \sim \tau \epsilon^3 \max{(p_i,q_i)}.
\ee
We are going to get these estimates by calling
$\EstProb{p}{\{i\}}{M}$ and $\EstProb{q}{\{i\}}{M}$.
The number of queries $M$ has to be chosen sufficiently large such that conditions Eq.~(\ref{BHMT1}) are satisfied for precision
$\delta$ defined in Eq.~(\ref{delta}) and error probability determined by Eq.~(\ref{error_bound4}), that is,
\be
\label{omega}
\omega \sim \tau n^{-1}  \sim \tau^2 \epsilon^2.
\ee
It leads to the condition
\be
\label{error_bound6}
M\ge \Omega\l(\frac{1}{\tau^3 \epsilon^5 \max{(\sqrt{p_i},\sqrt{q_i})}}\r).
\ee
Recall that we are interested in the case when $i$ is good. In this case
$\max{(p_i,q_i)}\ge \tau/(3nN) \sim N^{-1} \tau^2 \epsilon^2$.
Therefore Eq.~(\ref{error_bound6}) is satisfied whenever
\[
M\ge \Omega\l(\frac{1 \, \sqrt{N}}{\tau^4 \epsilon^6}\r).
\]
\end{proof}
Theorem~\ref{thm:statdif} follows directly from Lemma~\ref{lemma:statdif} since
$\EstDist{p}{q}{\epsilon}{\tau}$ makes $O(\sqrt{N})$ queries to the
quantum oracles generating $p$ and $q$.

%%%%%%%%%%%%%%%%%%%%%%%%%%%%%%%%%%%%%%%%%%
%%%%%%%%%%%%%%%%%%%%%%%%%%%%%%%%%%%%%%%%%%
\section{Quantum algorithm for testing Uniformity}
\label{sec:uniform}
%%%%%%%%%%%%%%%%%%%%%%%%%%%%%%%%%%%%%%%%%%
%%%%%%%%%%%%%%%%%%%%%%%%%%%%%%%%%%%%%%%%%%
In this section we prove Theorem~\ref{thm:uniform}. Let $p\in \calD_N$ be an unknown distribution specified by an oracle.
We are promised that either $p$ is the uniform distribution, or $p$ is {\em $\epsilon$-nonuniform}, that is, the $L_1$-distance
between $p$ and the uniform distribution is at least $\epsilon$. The algorithm described below is based on the following simple
observation. Choose some integer $M\ll N$ and
let  $S=(i_1,\ldots,i_M)$ be a list of  $M$ independent samples drawn from the distribution $p$.
Define a random variable $p_S=\sum_{a=1}^M p_{i_a}$.
It coincides with the total probability of all elements in $S$ unless $S$ contains a collision
(that is, $i_a=i_b$ for some $a\ne b$).
 The characteristic property of the uniform distribution is that $p_S=M/N$ with certainty.
 On the other hand, we shall see that for any $\epsilon$-nonuniform distribution $p_S$ takes values
greater than $(1+\delta)M/N$ for some constant $\delta>0$ depending on $\epsilon$ with a non-negligible probability. This observation suggests the
following algorithm for testing uniformity (the constants $K$ and $M$  below will be chosen later).

%SBB: now we reject if S has a collision
\begin{center}
%\begin{figure}[h]
\fbox{\parbox{16cm}{ $\UTest{p}{K}{M}{\epsilon}$
\noindent

Let $S=(i_1,\ldots,i_M)$ be a list of $M$ independent samples drawn from $p$.

Reject unless all elements in $S$ are distinct.

Let $p_S=\sum_{a=1}^M p_{i_a}$ be the total probability of elements in $S$.

Let $\tilde{p}_S$ be an estimate of $p_S$ obtained using $\EstProb{p}{S}{K}$.

If $\tilde{p}_S> (1+\epsilon^2/8)M/N$ then reject. Otherwise accept.
}}
%\end{figure}
\end{center}
This procedure will need to be repeated several times to achieve the desired bound on the error
probability, see the proof of Theorem~\ref{thm:uniform} below.

The main technical result of this section is the following lemma.
%SBB: simplified
\begin{lemma}
\label{lemma:uniform}
Let $p\in \calD_N$ be an $\epsilon$-nonuniform distribution.
Let $S=(i_1,\ldots,i_M)$ be a list of $M$ independent samples drawn from $p$, where
\be
\label{M}
M^3 =\frac{32 N}{\epsilon^4}.
\ee
Let $p_S=\sum_{a=1}^M p_{i_a}$ and $\alpha= 2^8 \epsilon^{-4}$.
Then
\be
\label{pAbound}
\prob{ p_S \ge (1+\epsilon^2/2) \frac{M}{N}} \ge \frac12\, \exp{(-\alpha)}.
\ee
\end{lemma}
Theorem~\ref{thm:statdif} follows straightforwardly  from  the above lemma and Theorem~\ref{thm:BHMT}.
\begin{proof}[Proof of Theorem~\ref{thm:statdif}]
Let $M$ be chosen as in Eq.~(\ref{M}) and
\[
K=c \frac{e^{\alpha} N^{1/3}}{\epsilon^{4/3}},
\]
where  $c=O(1)$ is a constant to be chosen later. Consider the following algorithm:
%SBB: updated
\begin{center}
\fbox{\parbox{18cm}{
Perform $L=4\exp{(\alpha)}$ independent tests  $\UTest{p}{K}{M}{\epsilon}$.
If at least one of the tests outputs `reject' then reject. Otherwise accept.}}
\end{center}

Let us show that this algorithm rejects any $\epsilon$-nonuniform distribution with
probability at least $2/3$ and accepts the uniform distribution with probability at least $2/3$.

\noindent
{\em Part 1: Any $\epsilon$-nonuniform distribution is rejected with high probability.}
 Let $P_{s}$ be the probability that for at least one of the {\bf UTest}s one has
\be
\label{pA_condition}
p_S \ge (1+\epsilon^2/2) \frac{M}{N}
\ee
Using Lemma~\ref{lemma:uniform} we conclude that
\be
\label{ps}
P_{s} \ge 1-\left(1-\frac1{2e^{\alpha}}\right)^{4e^\alpha} \ge 1-e^{-2} \ge \frac56.
\ee
In what follows we shall focus on a single test $\UTest{p}{K}{M}{\epsilon}$ that satisfies
Eq.~(\ref{pA_condition}) and show that it outputs `reject' with high probability.
Indeed, let $S$ be the sample list generated by this {\bf UTest}.
If $S$ contains a collision, the test outputs `reject'. Otherwise $p_S$ coincides
with the total probability of all elements in $S$.
The test outputs `reject' whenever $p_S$ is estimated with a precision
\be
\label{deltapA}
\delta =p_S \frac{\epsilon^2}{4}.
\ee
 In this case
\[
\tilde{p}_S \ge \l(1-\frac{\epsilon^2}{4}\r) p_S \ge
\l(1-\frac{\epsilon^2}{4}\r)\l(1+\frac{\epsilon^2}{2}\r)\frac{M}{N}
 > \l(1+\frac{\epsilon^2}{8}\r) \frac{M}{N}.
\]
(Here we assumed for simplicity that $\epsilon\le 1$.)
Suppose we want the {\bf UTest} to output `reject' with probability at least $5/6$.
Applying Eq.~(\ref{BHMT1}) with $\delta$ defined in Eq.~(\ref{deltapA}) and
$\omega=1/6$ we arrive at
\be
\label{Kchoice1}
K\ge \frac{c}{\epsilon^2 \sqrt{p_S}}
\ee
for some constant $c=O(1)$. Using Eq.~(\ref{pA_condition}) it suffices to choose
\be
\label{Kchoice2}
K = O\l(  \frac{\sqrt{N}}{\epsilon^2 \sqrt{M}} \r)
= O\l( \frac{ N^{1/3}}{\epsilon^{4/3}}\r)
\ee
Summarizing, if $p$ is an $\epsilon$-nonuniform distribution  it will
be rejected
with probability at least $(5/6)^2 \ge 2/3$.

\vspace{3mm}

%SBB: now we have to worry that the uniform distribution can be rejected because S contains a  collision
\noindent
{\em Part 2: The uniform distribution is accepted with high probability.}
Note that the uniform distribution can be rejected for two possible reasons: (i) for some {\bf UTest}
the sample list $S$ contains a collision; (ii) for some {\bf UTest} the estimate $\tilde{p}_S$
is sufficiently large, $\tilde{p}_S>(1+\epsilon^2/8)\, M/N$. We analyze these two possible sources of errors below.

(i)  For any fixed {\bf Utest}  let $S=(i_1,\ldots,i_M)$ be a list of $M$ samples drawn from $p$.
Let $C$ be the number of collisions in $S$, that is,  the number of pairs
$1\le a<b\le M$ such that $i_a=i_b$. Then,
\[
\EE(C)={M\choose 2} \sum_{i=1}^N p_i^2  \le \frac{M^2}{2N}.
\]
 Markov's inequality implies that $\prob{C\ge 1}\le \EE(C)\le M^2/(2N)$.
Then the probability that at least one
of the {\bf UTests} will find a
collision can be bounded using the union bound
as
\[
P_c\le  \frac{L M^2}{2N} = O\l( \frac1{N^{1/3}}\r)
\]
since we have chosen $M=O(N^{1/3})$ and $L=O(1)$. Thus the error probability associated with finding collisions
can be neglected.

(ii) Let $\tilde{p}_S$ be the estimate of $p_S$ obtained in some fixed {\bf UTest}.
Since $p_S=M/N$ with certainty, the test
outputs `accept' whenever the estimate $\tilde{p}_S$ returned by $\EstProb{p}{S}{K}$ satisfies
$|\tilde{p}_S-p_S|\le \delta$, where
\be
\label{deltapA1}
\delta=\frac{\epsilon^2 M}{8N}.
\ee
Since the total number of {\bf Utest}s is $L=4e^\alpha$, we would like the estimate $\tilde{p}_S$
to have precision $\delta$ with error probability $\omega \leq
\frac{1}{12}e^{-\alpha}$.
Applying Eq.~(\ref{BHMT1}) with $\delta$, $\omega$ defined above and taking into account that
$p_S = M/N$, we find that we can take the number of queries $K$ to be
\be
\label{Kchoice3}
K =  O\l(\frac{\sqrt{p_S}}{\omega \delta} \r)=
O\l( \frac{e^\alpha N^{1/3}}{\epsilon^{4/3}}\r).
\ee
It remains to choose the largest of Eq.~(\ref{Kchoice2}) and Eq.~(\ref{Kchoice3}).
\end{proof}

In the rest of this section we prove Lemma~\ref{lemma:uniform}.
We shall adopt notations introduced in the statement of Lemma~\ref{lemma:uniform}, that is,
the number of samples $M$ is defined by
\[
M^3=32\epsilon^{-4} N,
\]
$\alpha \equiv 2^8 \epsilon^{-4}$,
$S=(i_1,\ldots,i_M)$ is  a list of $M$ independent samples drawn from $p$,
and $p_S=\sum_{a=1}^M p_{i_a}$.

\begin{dfn}
An element $i\in [N]$ is called big iff $p_i>1/(2M^2)$.
\end{dfn}
Define the set $\bg\subset [N]$ of all big elements and their total probability:
\be
\label{big}
\bg=\{ i\in [N]\, : \, p_i >1/(2M^2)\}, \quad
\wbg=\sum_{i\in \bg} p_i.
\ee

We shall start in see subsection~\ref{subs:no} by proving
Lemma~\ref{lemma:uniform} for the special
case when $p$ has no big elements.  The proof is based on Chebyshev's
inequality.  Then we shall leverage this result in
subsection~\ref{subs:few} to show that
distributions with a few big elements (small $\wbg$) also satisfy
Lemma~\ref{lemma:uniform}. Finally in subsection~\ref{subs:many}, we
shall treat distributions with many big elements (large
$\wbg$) using a completely different technique.

%SBB: the proof is simplified
\subsection{Proof of Lemma~\ref{lemma:uniform}: no big elements}
\label{subs:no}
%%%%%%%%%%%%%%%%%%%%%%%%%%%%%%%%%%%%%%%%%%%%%
% no big elements
%%%%%%%%%%%%%%%%%%%%%%%%%%%%%%%%%%%%%%%%%%%%%

\begin{lemma}[\bf No big elements] Suppose $p\in \calD_N$ is $\epsilon$-nonuniform
and has no  big elements.  Then
\be
\label{nobig_bound}
\prob{ p_S \ge \l(1+\frac{\epsilon^2}{2}\r) \frac{M}{N}} \ge \frac34.
\ee
\label{lemma:nobig}
\end{lemma}
\begin{proof}

One can easily check that
\be
\label{EV}
\EE{(p_S)} = M  \la p|p\ra,  \quad \Var{(p_S)} =M\left(\sum_{i=1}^N p_i^3 -\la p|p\ra^2\right).
\ee

\begin{prop}
\label{prop:2norm}
Suppose $p\in \calD_N$ is $\epsilon$-nonuniform.   Then
\be
\label{2norm}
\la p|p\ra \ge \frac{1+\epsilon^2}{N}.
\ee
\end{prop}
\begin{proof}
Let $u$ be the uniform distribution.
Then $\epsilon\le \|p-u \|_1 \le \sqrt{N} \, \|p-u\|_2 = \sqrt{N} \,
\sqrt{\la p|p\ra - N^{-1}}$ which gives the desired bound.
\end{proof}

Using the proposition and the assumption that $p$ has no big elements we get
\be
\label{bounds2}
\EE{(p_S)}\ge \frac{M}{N}(1+\epsilon^2), \quad \Var{(p_S)} \le M
\|p\|_\infty \la p|p\ra \le \frac1{2M} \la p|p\ra.
\ee

 Chebyshev's inequality implies that
\be
\label{Che}
\prob{|p_S-\EE{(p_S)}| \ge t \EE{(p_S)}}\le \frac{\Var{(p_S)}}{\EE{(p_S)}^2\,  t^2}.
\ee
Assuming for simplicity that $\epsilon^2\le 1/3$ we can use the bound
$(1+\epsilon^2)^{-1} \le 1-3\epsilon^2/4$ and thus
\[
\prob{p_S\le \l(1+\frac{\epsilon^2}{2}\r)\cdot \l(\frac{M}{N}\r)}
 \le \prob{p_S \le \EE{(p_S)} \, \frac{(1+\epsilon^2/2)}{(1+\epsilon^2)}}
\le  \prob{p_S \le (1-\epsilon^2/4) \EE{(p_S)}}.
\]
Using Eq.~(\ref{Che}) with $t=\epsilon^2/4$ and Eqs.~(\ref{EV},\ref{bounds2}) we arrive at
\[
\prob{p_S\le (M/N)(1+\epsilon^2/2)} \le \frac{\la p|p\ra}{2M} \frac1{M^2 \la p|p\ra^2 t^2}\le \frac{8N}{M^3 \epsilon^4} \le \frac14
\]
since $\la p|p\ra \ge N^{-1}$ for any distribution $p\in \calD_N$ and since we have chosen $M^3=32\epsilon^{-4} N$.
%Therefore the probability that $S$ has no collisions {\em and} the simplified random variable $p_S$
 %defined in Eq.~(\ref{p_S'}) satisfies  $p_S\ge (1+\epsilon^2/2)M/N$
 %is at least $1/2$.
\end{proof}

\subsection{Proof of Lemma~\ref{lemma:uniform}: a few big elements}
\label{subs:few}
%%%%%%%%%%%%%%%%%%%%%%%%%%%%%%%%%%%%%%%%%%%%%
% a few big elements
%%%%%%%%%%%%%%%%%%%%%%%%%%%%%%%%%%%%%%%%%%%%%

\begin{lemma}[\bf A few big elements]
 Suppose $p\in \calD_N$ is $\epsilon$-nonuniform
and has only a few  big elements such that
\be
\label{fewbig}
\wbg\le \frac{\alpha}{M}, \quad \alpha \equiv 2^8 \epsilon^{-4}.
\ee
Then
\be
\label{fewbig_bound}
\prob{ p_S \ge (1+\epsilon^2/2) \frac{M}{N}} \ge \frac12\,  \exp{(-\alpha)}.
\ee
\label{lemma:fewbig}
\end{lemma}

\begin{proof}
Let $S=(i_1,\ldots,i_M)$ be a list of $M$ samples drawn from $p$.
We can get a constant lower bound on the probability that $S$ contains no big elements:
\bea
\label{prob_nobig}
\prob{S\cap \bg =\emptyset} &=& (1-\wbg)^M \approx \exp{(-M\wbg)}
\ge e^{-\alpha}.
\eea
(Strictly speaking, one gets a lower bound $e^{-\alpha}(1-o(1))$.)
It suffices to show that $p_S\ge (1+\epsilon^2/2) M/N$ with probability at least $1/2$ conditioned on $S$
having no big elements.

The conditional distribution of the random variable $p_S$ given that $S$ contains no
big elements can be obtained by setting the probability of all big elements to zero and
renormalizing $p$ by a factor $(1-\wbg)^{-1}$. In other words,  we can repeat all arguments of
Lemma~\ref{lemma:nobig} if we replace $p$ by a new distribution $p'\in \calD_N$ such that
\be
\label{pprime}
p'_i=\left\{ \ba{rcl} \frac{p_i}{(1-\wbg)} &\mbox{if} & i\notin \bg, \\
0 &\mbox{if} & i\in \bg.\\
\ea\right.
\ee
We have to check that $p'$ is also $\epsilon$-nonuniform.
\begin{prop}
The distribution $p'$ is $\epsilon'$-nonuniform, where $\epsilon'\ge \epsilon-O(N^{-1/3})$.
\end{prop}

\begin{proof}
\[
\|p-p'\|_1 =\sum_{i\in \bg} p_i + \sum_{i\notin \bg} \left[ (1-\wbg)^{-1}-1\right] p_i  \le \wbg  + \frac{\wbg}{(1-\wbg)}
=O(N^{-1/3}).
\]
Let $u$ be the uniform distribution.
Using the triangle inequality we get
\[
\|p'-u\|_1 \ge \| p-u\|_1 -\|p-p'\|_1 \ge \epsilon -O(N^{-1/3}).
\]
\end{proof}
To simplify notations
we shall neglect the correction of order $N^{-1/3}$ and assume that $p'$ is $\epsilon$-nonuniform.
 By construction,
\[
\|p'\|_\infty \le \frac1{(1-\wbg) 2M^2}=1/(2M^2)  +O(N^{-1}).
\]
Neglecting the correction of order $N^{-1}$ we can assume that $p'$ has no big elements.
Then Lemma~\ref{lemma:nobig} implies that
$p_S'\ge (1+\epsilon^2/2) M/N$ with probability at least $3/4$.
Combining it with Eq.~(\ref{prob_nobig}) we arrive at Eq.~(\ref{fewbig_bound}).
\end{proof}

\subsection{Proof of Lemma~\ref{lemma:uniform}: many big elements}
\label{subs:many}

%%%%%%%%%%%%%%%%%%%%%%%%%%%%%%%%%%%%%%%%%%%%%
% many big elements
%%%%%%%%%%%%%%%%%%%%%%%%%%%%%%%%%%%%%%%%%%%%%

\begin{lemma}[\bf Many big elements] Suppose $p$ is $\epsilon$-nonuniform and has many big elements
such that
\be
\label{case1}
\wbg>\frac{\alpha}{M}, \quad \alpha\equiv 2^8 \epsilon^{-4}.
\ee
Then
\be
\prob{p_S \ge 2\frac{M}{N}} \ge \frac12.
\ee
\end{lemma}
\begin{proof}
Let $S=(i_1,\ldots,i_M)$ be a list of $M$ independent samples drawn from $p$.
Since each big element contained in $S$ contributes at least $1/(2M^2)$ to $p_S$,
the inequality $p_S\ge 2M/N$ is satisfied whenever $S$ contains at least $n$  big elements
where
\[
\frac{n}{2M^2} \ge \frac{2M}{N}.
\]
Since  $M^3=2^5 \epsilon^{-4}N$, we can choose
\be
\label{nb}
n=2^7 \epsilon^{-4}= \alpha/2.
\ee
The total number of samples $a\in [M]$ such that $i_a$ is big can be represented as
$\xi=\sum_{i=1}^M \xi_i$, where
$\xi_i\in \{0,1\}$ is a random variable such that $\xi_i=1$ iff $i$ is a big element.
Note that $\EE(\xi)=M\wbg >\alpha$.
 Using Chebyshev's inequality we get
\be
\prob{\xi< n} \le \prob{ |\xi-\EE(\xi) |\ge  \frac12\, \EE(\xi) } \le \frac{ 4\Var{(\xi)}}{\EE(\xi)^2} \le \frac{4}{\EE(\xi)}
\le \frac{4}{\alpha}\le \frac12.
\ee
\end{proof}

%%%%%%%%%%%%%%%%%%%%%%%%%%%%%%%%%%%%%%%%%%%%
%%%%%%%%%%%%%%%%%%%%%%%%%%%%%%%%%%%%%%%%%%%%
\section{Quantum algorithm for testing orthogonality}
\label{sec:orthog}
%%%%%%%%%%%%%%%%%%%%%%%%%%%%%%%%%%%%%%%%%%%%
%%%%%%%%%%%%%%%%%%%%%%%%%%%%%%%%%%%%%%%%%%%%
Consider distributions $p,q\in \calD_N$ and
let $S=(i_1,\ldots,i_M)$ be a list of $M$  independent samples drawn from $p$.
Let $A\subseteq [N]$ be the set of all elements that appear in $S$ at least once.
Define the {\em collision probability}
\[
q_A=\sum_{i\in A} q_i.
\]
Note that $q_A$ is a deterministic function of $A$, so the probability distribution of $q_A$
is determined by probability distribution of $A$ (which depends on $p$ and $M$).
For a fixed $A$ the variable $q_A$ is the probability  that a sample drawn from $q$ belongs to $A$.

Clearly if $p$ and $q$ are orthogonal then $q_A=0$ with probability $1$.
On the other hand, if $p$ and $q$ have a constant overlap, we will show that $q_A$
takes values of order $M/N$ with  constant probability. Specifically, we shall prove the following lemma.
%SBB: the old version had a mistake in Proposition 4. Fixing this mistake lead to extra factors 1/2 in several places.
\begin{lemma}
\label{lemma:orthog}
Consider a pair of distributions $p,q\in \calD_N$ such that
$\|p-q\|_1\le 2-\epsilon$. Let $q_A$ be a collision probability
constructed using $M$ samples.
Suppose  $M\ge 2^9\epsilon^{-2}$.
Then
\be
\label{q_A}
\prob{q_A\ge  \frac{\epsilon^3 M}{2^{11}N}} \ge \frac12.
\ee
\end{lemma}
It suggests the following algorithm for testing orthogonality.
\begin{center}
%\begin{figure}[h]
\fbox{\parbox{16cm}{ $\OTest{p}{q}{M}{K}$
\noindent

Let $S=\{i_1,\ldots,i_M\}$ be a list of $M$ independent samples drawn from $p$.

Let $A\subseteq [N]$ be the set of elements that appear in $S$ at least once.

Let $q_A=\sum_{i\in A} q_i$ be the total probability of elements in $A$ with respect to $q$.

Let $\tilde{q}_A$ be estimate of $q_A$ obtained using $\EstProb{q}{A}{K}$.

If $\tilde{q}_A\ge \frac{\epsilon^3 M}{2^{12} N}$ then reject. Otherwise accept.
}}
%\end{figure}
\end{center}
We note that if $q_A=0$ then $\tilde{q}_A=0$ with certainty (see
Theorem~\ref{thm:BHMT}) and so {\bf OTest}  accepts any pair of
orthogonal distributions with certainty.
Theorem~\ref{thm:orthog} is a direct consequence of the following lemma.
\begin{lemma}
Choose
\be
\label{KM}
M=K=O\l( \frac{N^{1/3}}{\epsilon}\r).
\ee
Then $\OTest{p}{q}{M}{K}$  rejects any distributions $p,q\in \calD_N$ such that $\|p-q\|_1\le 2-\epsilon$ with
probability at least $1/4$.
\end{lemma}
\begin{proof}
According Eq.~(\ref{q_A}), $q_A \geq \epsilon^3 M/(2^{11} N)$ with
probability $\geq 1/2$. When this holds, the algorithm rejects whenever
\[
|\tilde{q}_A - q_A|\le \frac{q_A}2
\]
since this implies $\tilde{q}_A \ge q_A/2 \ge \epsilon^3 M/(2^{12} N)$.
Applying Theorem~\ref{thm:BHMT}  with precision $\delta=q_A/2$ and
error probability $\omega =1/2$, we find (according to
Eq.~(\ref{BHMT1})), that $K$ should be
\be
\label{Kchoice4}
K\ge \Omega\l(\frac{1}{\sqrt{q_A}}\r)
\ee
Taking into account Eq.~(\ref{q_A}) it suffices to choose
\[
K=\Omega\l(\frac{N^{1/2}}{\epsilon^{3/2} M^{1/2}}\r)
\]
to guarantee that {\bf Otest} outputs `reject' with probability at least $(1/2)\cdot (1/2) =1/4$.
Minimizing the total number of queries $K+M$ we arrive at Eq.~(\ref{KM}).
\end{proof}

In the rest of this section we prove Lemma~\ref{lemma:orthog}.
\begin{proof}
Begin
by defining two sets of indices:
\bal B & \equiv \{ i : q_i < \frac{\eps}{4} p_i\} \\
C & \equiv \{ i : p_i \leq \frac{\eps}{32}N^{-1}\}
\eal
Let $B^c,C^c$ denote the complements of $B$ and $C$ respectively.
We will prove that
\be \prob{|A \cap B^c \cap C^c| \geq \frac{\eps}{16}M} \geq 1/2
\label{eq:ABC-bound},\ee
 which will imply the Lemma since
\be
q_A\ge \sum_{i\in A\cap B^c\cap C^c} q_i \ge \frac{\epsilon}4\;  \sum_{i\in A\cap B^c\cap C^c } p_i\ge
 \frac{\epsilon^2}{2^7 N} \, |A\cap B^c\cap C^c|.
\ee

First, we show that $|A\cap B|$ is likely to not be too big.  Observe
that $q_B < \frac{\eps}{4} p_B \leq \frac{\eps}{4}$.
Next use the fact that $\frac{1}{2}\|p-q\|_1 = \max_{U\subset [N]} p_U-q_U \leq
1-\frac{\eps}{2}$ to bound $p_B \leq 1- \frac{\eps}{2} +
\frac{\eps}{4} = 1- \frac{\eps}{4}$.  Now we state a Chernoff-Hoeffding
bound.
\begin{lemma}\label{lem:chernoff}
Let $X_1,\ldots,X_M$ be independent $0,1$ random
  variables with $X\equiv \sum_{i=1}^M X_i$.  Then
  for any $\delta>0$,
 \be
 \prob{X \geq \expect{X} + M\delta} \leq \exp(-2M\delta^2).
 \ee
\end{lemma}

Recall that $A$ consists of the unique elements of
$S=\{i_1,\ldots,i_M\}$.  For $j=1,\ldots,M$, define $X_j = 1$ if
$i_j\in B$ and $X_j=0$ if not.  Then $|A\cap B| \leq \sum_{j=1}^M
X_j$, with the possibility of an inequality in case there are
repeats.  We can now use \lemref{chernoff}
with $\expect{X_j}=p_B\leq 1-\eps/4$ and $\delta=\epsilon/8$
to prove that
\be \prob{|A \cap B| \geq  \l(1-\frac{\eps}{8}\r)M} \leq
\exp\l(-2M\l(\frac{\eps}{8}\r)^2\r) = \exp\l(-\frac{M\eps^2}{32}\r)
\label{eq:AB-bound}.\ee

Next, we observe that $p_C \leq {\eps}/{32}$.  We can use the same
method to show that $|A\cap C|$ is likely to not be too big.  This
time we define $X_j = 1$ iff $i_j\in C$, so that
$|A\cap C| \leq \sum_{j=1}^M X_j$ and
$\expect{X_j}=p_C \leq \eps/16$. Setting $\delta=\epsilon/32$ we get
\be\prob{|A\cap C| \geq \frac{\eps}{16}M} \leq
\exp\l(-\frac{M\eps^2}{2^9}\r)
\label{eq:AC-bound}.\ee

When $M\geq 2^{9}/\eps^2$, we can combine \eq{AB-bound} and \eq{AC-bound}
to find that with probability $\geq 1/2$, both
$|A \cap B^c| \geq \frac{\eps}{8}M$ and
$|A\cap C^c|\geq (1-\frac{\eps}{16})M$.  Thus
$|A \cap B^c \cap C^c| \geq \frac{\eps}{16}M$ with probability at
least $1/2$.   This establishes \eq{ABC-bound}, and completes the
proof of the lemma.
\end{proof}

%%%%%%%%%%%%%%%%%%%%%%%%%%%%%%%%%%%%%%%%%%%%
%%%%%%%%%%%%%%%%%%%%%%%%%%%%%%%%%%%%%%%%%%%%
\section{Lower bounds}
\label{sec:lower}
%%%%%%%%%%%%%%%%%%%%%%%%%%%%%%%%%%%%%%%%%%%%
%%%%%%%%%%%%%%%%%%%%%%%%%%%%%%%%%%%%%%%%%%%%

\subsection{Sampling vs query complexity}
Let $p\in \calD_N$ be any distribution and $O\, :\, [S] \to [N]$ be an oracle generating $p$.
Recall that $p_i$ coincides with the fraction of inputs $s\in [S]$ such that $O(s)=i$.
 It does not matter which particular inputs $s$ are mapped to $i$. The
 only thing that matters is the number of such inputs. Therefore one
 can choose an arbitrary permutation of inputs $\sigma\, :\, [S]\to
 [S]$
and construct a new oracle $O'=O\circ \sigma$ that generates the same distribution $p$.
We shall see below that if a classical testing algorithm $\calA$ gives a correct answer with high probability for any choice of $S$ and $\sigma$
then $\calA$ cannot take any advantage from making adaptive queries to $\calO$. Let us transform $\calA$
into a `sampling' algorithm $\calA_{s}$ such that each query made in $\calA$ is replaced by a random query drawn
from the uniform distribution on $[S]$.

\begin{lemma}
\label{lemma:adapt}
Let $\calA$ be any classical testing algorithm
and $p\in \calD_N$ be some distribution such that $\calA$ accepts (rejects) $p$
with probability at least $2/3$ for any oracle $O\, : \, [S]\to [N]$ generating $p$.
Then the corresponding sampling algorithm $\calA_s$ accepts (rejects) $p$
with probability at least $2/3$.
\end{lemma}
\begin{proof}
 Let $P_{acc}(\sigma)$ be a probability that $\calA$
accepts while interacting with the oracle $O\circ \sigma$, where $\sigma$ is
a permutation on $[S]$.  Without loss of generality $P_{acc}(\sigma)\ge 2/3$
for all $\sigma$. It implies that the average acceptance probability
\be
\label{average_acc}
P_{acc}=\frac1{S!} \sum_{\sigma} P_{acc}(\sigma) \ge \frac23.
\ee
An execution of the algorithm $\calA$ can
be represented by a history of queries $Q=(s_1,\ldots,s_T)\in [S]^{\times T}$.
Let $P(Q)$ be a probability that an execution of $\calA$ leads to a history $Q$.
 We can assume without loss of generality that the output of $\calA$ (accept or reject)
is a deterministic function of $Q$. Let $\Omega_{acc}$ be a set of histories $Q$
that make $\calA$ to accept. We have $P_{acc}(\sigma)=\sum_{Q \in \Omega_{acc}} P(\sigma^{-1} Q)$,
where
\[
\sigma^{-1} Q \equiv (\sigma^{-1}(s_1),\ldots,\sigma^{-1}(s_T)),
\]
 and thus
\[
P_{acc}= \sum_{Q\in \Omega_{acc}}  \frac1{S!} \sum_{\sigma } P(\sigma^{-1} Q)\ge \frac23.
\]
Let $\bar{P}(Q)=\bbE(P(\sigma^{-1} Q))$ where $\sigma$ is drawn from the uniform distribution.
Let $U(Q)$ be the uniform distribution on the set $[S]^{\times T}$.
We claim that
\be
\label{PU_bound}
\| \bar{P} - U \|_1 =O(TS^{-1}).
\ee
Assume without loss of generality that all queries in $Q$ are different.
Then
\[
\bar{P}(Q)=\frac{(S-T)!}{S!} =S^{-T}(1+O(T^2/S)).
\]
A probability that a history drawn from the uniform distribution contains two or more equal queries
can be bounded by $O(T^2/S)$ and thus we arrive at Eq.~(\ref{PU_bound}).
Therefore in the limit $S\to \infty$ the acceptance probability is at least $2/3$
if $Q$ is drawn from the uniform distribution. But this implies that the sampling algorithm $\calA_s$
accepts $p$ with probability at least $2/3$.
\end{proof}

\subsection{Reduction from the Collision Problem to testing Orthogonality}
\label{subs:collision}

One can get lower bounds on the query complexity of testing Orthogonality using the lower bounds
for the Collision problem~\cite{AS-collision-04}. Indeed, let $H\, : \, [N]\to [3N/2]$ be an oracle
function such that either $H$ is one-to-one (yes-instance) or $H$ is two-to-one (no-instance).
The Collision Problem is to decide which one is the case.
It was shown by Aaronson and Shi~\cite{AS-collision-04}
that the quantum query complexity of the Collision problem is $\Omega(N^{1/3})$.
Below we show that the Collision problem can be reduced to testing Orthogonality\footnote{In order to apply the lower bound
proved in~\cite{AS-collision-04} one has to choose the range of $H$ of size $3N/2$ rather than $N$ which would be more natural.}.
It implies that testing Orthogonality requires $\Omega(N^{1/2})$ queries classically and $\Omega(N^{1/3})$ queries
quantumly.

Indeed, choose a random permutation $\sigma\, : \, [N]\to [N]$ and define functions
$O_p,O_q\, : \, [N/2]  \to [3N/2]$ by restricting the composition $H\circ \sigma$ to the subsets
of odd and even integers respectively:
\[
O_p(s)=H(\sigma(2s-1)), \quad O_q(s)=H(\sigma(2s)),\quad s\in [N/2].
\]
 For any yes-instance (i.e. $H$ is one-to-one), the distributions $p,q\in
\calD_{3N/2}$ generated by $O_p$ and $O_q$ are
 uniform distributions on some pair of  disjoint subsets of $[3N/2]$; that is,
$p$ and $q$ are orthogonal.

We need to show that for any no-instance ($H$ is two-to-one) the
distance  $\|p-q\|_1$ takes values
smaller than $2-\epsilon$  with a sufficiently high probability for some
constant $\epsilon$.
\begin{lemma}
\label{lemma:col_low_bound}
Let $H\, : \, [N]\to [3N/2]$ be any two-to-one function. Let $\sigma\, : \, [N]\to [N]$ be
a random permutation drawn from the uniform distribution. Then
\[
\prob{\| p-q\|_1 \le \frac74} \ge \frac12.
\]
\end{lemma}
\begin{proof}
Given the promise on $H$ we can
define a perfect matching $\calM$ on the set $[N]$
(considered as a complete graph with $N$ vertices)
such that
$H(u)=H(v)$ iff $u$ and $v$ are matched. Let $\calM_\sigma=\sigma^{-1}\circ \calM$.
Clearly, $\calM_\sigma$ is a random perfect matching on $[N]$ drawn from the uniform distribution
on the set of all perfect matchings. Let  $(u,v)\in \calM_\sigma$ be some pair of matched vertices
and $w=H(\sigma(u))=H(\sigma(v))$. Note that if $u$ and $v$ have different parity then
$p_w=q_w=2/N$. On the other hand, if $u$ and $v$ have the same parity then $p_w=4/N$,
$q_w=0$ or vice verse.  Thus
\be
\label{matching|parity}
\|p-q\|_1 = 2-\frac4{N} \#\{ (u,v)\in \calM_{\sigma} \,: \, \mbox{$u$ and $v$ have different parity}\}.
\ee
A nice property of the uniform distribution on the set of perfect matchings on $[N]$ is that a
conditional distribution given that $(u,v)\in \calM_{\sigma}$ is the uniform distribution
on the set of perfect matchings on $[N]\backslash \{u,v\}$. Thus we can generate $\calM_\sigma$ using the
following algorithm.
Let $U\subseteq [N]$ be the set of all unpaired vertices (in the beginning $U=[N]$). Let $U_{even}$ and $U_{odd}$
be the subsets of all even and all odd integers in $U$.
The algorithm starts from an empty matching $\calM_\sigma =\emptyset$.
Suppose at some step of the algorithm we have some matching $\calM_\sigma$
and some sets of unpaired vertices $U=U_{even}\cup U_{odd}$.
If $|U_{even}|\ge |U_{odd}|$ choose
a random vertex $u\in U_{odd}$.
If $|U_{even}|<|U_{odd}|$ choose a random vertex $u\in U_{even}$.
Pair $u$ with a random vertex $v\in U\backslash \{u\}$ and
update
\[
\calM_\sigma \to \calM_\sigma \cup \{u,v\}, \quad U\to U\backslash \{u,v\}
\]
with the corresponding update for $U_{even}$ and $U_{odd}$.
After $N/2$ steps of the algorithm we generate a random uniform $\calM_\sigma$.

By construction, at each step of the algorithm we pair a vertex $u$ to a vertex $v$ with the opposite parity
with probability at least $1/2$.  Thus the probability $P(k)$ of having a matching $\calM_\sigma$
with less than $k$ pairs having opposite parity is
\[
P(k)\le \sum_{i=0}^k {N/2 \choose i} 2^{-\frac{N}2 + k} \le 2^{\frac{N}{2} \left[ H(x) + x-1 + o(1) \right] },
\]
where $x=2k/N$. One can check that $H(x)+x-1<0$ for $x\le 1/8$ and thus
$P(N/16) \le 1/2$ for sufficiently large $N$.
Thus Eq.~(\ref{matching|parity}) implies that $\|p-q\|_1\le 2-1/4=7/4$ with probability at least $1/2$.
\end{proof}

\subsection{Classical lower bound for testing Uniformity}
\label{subs:ulower}

%SBB:
In this section we prove that classically testing Uniformity requires  $\Omega(N^{1/2})$.
A proof uses the machinery developed by Valiant in~\cite{Val-symmetric-08}. Valiant's techniques
apply to testing {\em symmetric} properties of distributions, that is, properties that are invariant under relabeling of elements in the domain of a distribution. Clearly, Uniformity is a symmetric property.

We shall need two technical tools from~\cite{Val-symmetric-08}, namely, the Positive-Negative Distance lemma and Wishful Thinking theorem (see Theorem~4 and Lemma~3 in~\cite{Val-symmetric-08}).  Let us start from introducing some notations. Let $p\in \calD_N$
be an unknown distribution and $S=(i_1,\ldots,i_M)$ be a list of $M$ independent samples drawn from $p$.
We shall say that $S$ has a collision of order $r$ iff some element $i\in [N]$ appears in $S$ exactly $r$ times.
Let $c_r$ be the total number of collisions of order $r$, where $r\ge 1$. A sequence  of integers $\{c_r\}_{r\ge 1}$ is called a {\em fingerprint} of $S$. Define a probability distribution $D^M_p$ on a set of fingerprints as follows: (1) draw $k$ from the Poisson distribution
${\mathrm {Poi}}(k)=e^{-M} M^k/k!$. (2) Generate a list $S$ of $k$ independent samples drawn from $p$. (3) Output a fingerprint of $S$.

An important observation made in~\cite{Val-symmetric-08} is that a fingerprint  contains all relevant information
about a sample list as far as testing symmetric properties is concerned. Thus without loss of generality, a testing algorithm has to make its decision by looking only on a fingerprint of a sample list. Applying Positive-Negative Distance lemma from~\cite{Val-symmetric-08} to testing Uniformity we get the following result.
\begin{lemma}[\bf \cite{Val-symmetric-08}]
\label{lemma:Val1}
Let $u$ be the uniform distribution on $[N]$ and $p\in \calD_N$ be any distribution such that
$\|p-u\|_1\ge 1$. If for some integer $M$
\be
\label{Val1}
\| D^M_p - D^M_u \|_1 <\frac1{12}
\ee
then Uniformity is not testable in $M$ samples.
\end{lemma}
The second technical tool is a usable upper bound on the distance between the distributions of fingerprints.
For any integer $k$ define an $k$-th moment of $p$ as
\be
\label{Val2}
m_k(p)=\sum_{i=1}^N p_i^k.
\ee
Clearly $m_k(u)=N^{1-k}$ which is the smallest possible value of a $k$-th moment for distributions on $[N]$.
Applying Wishful Thinking theorem from~\cite{Val-symmetric-08} to testing Uniformity we get the following result.
\begin{lemma}[\bf \cite{Val-symmetric-08}]
Let $p\in \calD_N$ be any distribution such that $\|p\|_\infty \le \delta/M$ for some $\delta>0$. Then
\be
\label{Val3}
\| D^M_p - D^M_u\|_1 \le 40 \delta + 10 \sum_{k\ge 2} M^k\, \frac{ m_k(p)-N^{1-k}}{\lfloor{k/2}\rfloor! \sqrt{1+ M^k\, m_k(p)}}.
\ee
\end{lemma}
%SBB:
\begin{corol}
Uniformity is not testable classically in $32^{-1}\,  N^{1/2}$ queries.
\end{corol}
\begin{proof}
Consider a distribution
\[
p_i=\left\{ \ba{rcl} 2/N &\mbox{if} & 1\le i\le N/2, \\
0 && \mbox{otherwise}. \\
\ea
\right.
\]
Clearly $\|p-u\|_1=1$ and
\[
m_k(p)=2^{k-1} N^{1-k}.
\]
In particular, choosing $M=2^{-a}N^{1/2}$  we have
\[
M^k m_k(p)=2^{-k(a-1) -1}\, N^{1-\frac{k}2}\le 2^{-2a+1} \quad \mbox{for all $k\ge 2$}.
\]
Taking into account that
\[
\sum_{k\ge 2} \frac1{\lfloor{k/2}\rfloor!} \le 2(e-1)\le 4
\]
we can use Eq.~(\ref{Val3}) to infer that
\be
\label{Val4}
\| D^M_p - D^M_u\|_1 \le 40 \delta + 10\cdot 2^{-2a+3}.
\ee
Clearly, condition $\|p\|_\infty \le \delta/M$ can be satisfied for any constant $\delta>0$ and sufficiently large $N$.
Then Lemma~\ref{lemma:Val1} implies that Uniformity is not testable in $M$ samples
whenever $10\cdot 2^{-2a+3}<1/12$. It suffices to choose $a=5$.
Finally, Lemma~\ref{lemma:adapt} implies that Uniformity is not testable in $M$ queries in the oracle model.
\end{proof}

\vspace{5mm}

\noindent {\bf \large Acknowledgments}\\
We are grateful to Ronald de Wolf for numerous comments that helped to
improve the paper.  We would like to thank Sourav Chakraborty for
informing us about the results in~\cite{CFMW09}.  S.B.
thanks CWI for hospitality while this work was being done and gratefully
acknowledges financial  support from the DARPA QUEST program under contract
no. HR0011-09-C-0047.  A.W.H. was funded by EPSRC grant ``QIP IRC''
and is grateful to IBM and MIT for their hospitality while this work
was being done. A.H. received support from the xQIT Keck fellowship.

\bibliographystyle{abbrv}
%\bibliographystyle{plain}
%\bibliography{proptest}

\end{document}